\documentclass[11pt]{article}
\usepackage{amsmath,amsfonts,amsthm,alltt,amssymb}
\usepackage{graphicx,alltt}
\usepackage{xcolor}
\def\calU{\mathcal{U}}
\def\calV{\mathcal{V}}
\def\calW{\mathcal{W}}

\def\trace#1{\mathrm{Tr} \left(#1 \right)}

\def\abs#1{\left|#1  \right|}

\def\norm#1{\left\| #1 \right\|}

\newcommand\pphi{\boldsymbol{\mathit{\phi}}}

\let\ao\aa

\def\aa{\pmb{\mathit{a}}}

\newcommand\ee{\boldsymbol{\mathit{e}}}

\newcommand\qq{\boldsymbol{\mathit{q}}}

\newcommand\rr{\boldsymbol{\mathit{r}}}

\newcommand\vv{\boldsymbol{\mathit{v}}}

\newcommand\ww{\boldsymbol{\mathit{w}}}
\newcommand\xx{\boldsymbol{\mathit{x}}}

\newcommand\zz{\boldsymbol{\mathit{z}}}
\def\tt{\boldsymbol{\mathit{t}}}

\newcommand\zzhat{\boldsymbol{\widehat{\mathit{z}}}}

\def\union{\cup}

\newdimen\pIR
\newcommand\StevesR{{\rm I\kern\pIR R}}

\def\bvec#1{\mathbf{#1}}

\def\defeq{\stackrel{\mathrm{def}}{=}}
\def\setof#1{\left\{#1  \right\}}

\def\sizeof#1{\left|#1  \right|}

\newsavebox{\Codefig}

\def\lprod{{\ooalign{\hfil\raise .00ex\hbox{\scriptsize L}\hfil\crcr\mathhexbox20D}}}

\newtheorem{theorem}{Theorem}[section]

\newtheorem{lemma}[theorem]{Lemma}

\theoremstyle{definition}

\def\ttssp{2-2 Set Splitting Problem}
\def\ttsspt{(3,2-2) Set Splitting Problem}

\title{Hardness Results for Weaver's Discrepancy Problem%
\thanks{
This work was supported in part by NSF Grant CCF-1562041, ONR Award N00014-20-1-2335, and a Simons Investigator Award to Daniel Spielman
}}

\author{
Daniel A. Spielman \\ 
Yale University
\and 
Peng Zhang\\
Rutgers University
}

\begin{document}

\maketitle

\begin{abstract}
Marcus, Spielman and Srivastava~\cite{marcusSS15}
  solved the Kadison--Singer Problem by proving a strong form of Weaver's conjecture:
  they showed that for all $\alpha > 0$ and all lists of vectors
  of norm at most $\sqrt{\alpha}$ whose outer products sum to the identity,
  there exists a signed sum of those outer products with operator norm at most 
  $\sqrt{8 \alpha} + 2 \alpha.$
We prove that it is NP-hard to distinguish such a list of vectors for which
  there is a signed sum that equals the zero matrix from those in which every signed sum 
  has operator norm at least $\kappa \sqrt{\alpha}$, for some absolute constant $\kappa > 0.$
Thus, it is NP-hard to construct a signing that is a constant factor better than that guaranteed to exist.  
  
For $\alpha = 1/4$, we prove that it is NP-hard to distinguish whether there is a signed sum that equals the zero matrix from the case in which every signed sum has operator norm at least $1/4$.
\end{abstract}

\section{Introduction}

Implicit in Weaver's~\cite{weaver04} conjecture $KS_{2}$ is the following discrepancy problem:
  given vectors $\vv_{1}, \ldots, \vv_{n}$, find an $\xx \in \setof{\pm 1}^{n}$ minimizing
  the operator norm of $\sum_{i} \xx (i) \vv_{i} \vv_{i}^{*}$, where $\vv_{i}^{*}$ is the conjugate \footnote{%
  While all vectors in this paper have Real entries, Weaver's conjecture remains natural over the Complexes.}
  transpose of $\vv_{i}$.
Weaver proved that conjecture $KS_{2}$ implies a positive resolution of the Kadison--Singer Problem.
It is equivalent\footnote{
Weaver's statement is slightly different from this, but he proves it is equivalent to this in part b of his Theorem 2. See the Remarks section for some explanation.
}
to the statement 
  that there are
  constants $\alpha > 0$ and $\beta < 1$ such that for all vectors $\vv_{1}, \ldots, \vv_{n}$ that satisfy 
\[
  \norm{\vv_{i}}^{2} \leq \alpha, \text{for all $i$, and}
  \quad 
  \sum_{i = 1}^{n} \vv_{i} \vv_{i}^{*} = I,
\]
there exists a $\xx \in \setof{\pm 1}^{n}$ so that
\[
  \norm{\sum_{i} \xx (i) \vv_{i} \vv_{i}^{*} } \leq \beta.   
\]
Here, $\norm{\vv_{i}}$ refers to the standard Euclidean norm of a vector
  and the norm around the signed sum of outer products refers to the operator norm induced by the Euclidean vector norm:
\[
 \norm{M} = \max_{\tt : \norm{\tt} = 1} \norm{M \tt}.    
\]
Marcus, Spielman, and Srivastava~\cite{marcusSS15} solved the Kadison--Singer Problem by proving that Weaver's conjecture is true with $\beta = \sqrt{8 \alpha} + 2 \alpha$.
Their result was improved by Bownik, Casazza, Marcus, and Speegle \cite{BCMS19}, who reduced 
  the bound on $\beta$ to a little below $\sqrt{8 \alpha}$.

Neither of these results are accompanied by efficient algorithms, and many have wondered if there are efficient algorithms for choosing vectors $\xx$ that satisfy the conditions of Weaver's Conjecture.
The currently best known algorithm for constructing such an $\xx$
runs in time $O(2^{\sqrt[3]{n}/\alpha})$, and achieves $\beta$ arbitrarily close to $\sqrt{8 \alpha} + 2 \alpha$ \cite{AOSS18}.

In this paper, we prove that it is NP-hard to distinguish between the cases in which there exists an $\xx$ that makes the signed sum of outer products the all-0 matrix from the case in which all $\xx$ result in a signed sum with operator norm at least $\kappa \sqrt{\alpha}$, for some constant $\kappa > 0$.


Before stating our results in more detail, we introduce some notation that makes those statements compact.
Given a list of vectors $\calV = \vv_{1}, \ldots, \vv_{n}$
  and a vector $\xx \in \setof{\pm 1}^{n}$, we let $M(\calV, \xx)$
  denote the signed sum of outer products,
\[
  \sum_{i} \xx (i) \vv_{i} \vv_{i}^{*}.
\]
When just given the list of vectors $\calV$, we define the minimum achievable operator norm of such a signed sum of outer products to be 
\[
  W(\calV) = \min_{\xx \in \setof{\pm 1}^{n}} 
    \norm{M(\calV, \xx) }.
\]
We say that a list of vectors $\vv_{1}, \ldots, \vv_{n}$ is $\alpha$-Weaver if
$\sum_{i} \vv_{i} \vv_{i}^{*} = I$ and $\norm{\vv_{i}}^{2} \leq \alpha$ for all $i$.
In this notation, Weaver's conjecture $KS_{2}$ says that for some $\alpha > 0$
  and $\beta < 1$, every $\alpha$-Weaver list of vectors $\calV$ satisfies 
  $W(\calV) < \beta$.

We prove that there is a constant $\kappa > 0$ such that for every  $\alpha > 0$ it is NP-hard 
  to distinguish $\alpha$-Weaver lists of vectors with $W(\calV) = 0$
    from those for which $W(\calV) \geq \kappa \sqrt{\alpha}$.
As we know $W(\calV) \leq  \sqrt{8 \alpha}$, this result is optimal up to the constant $\kappa$.
Our proof depends on the NP-hardness of approximating Max 2-2 Set Splitting \cite{guruswami04,CGW05}.
The factor $\alpha$ can depend on the number of vectors: 
  we only require $\alpha \geq \Omega(n^{-1/2})$. 
We begin by showing that for $1/4$-Weaver vectors $\calV$, it is NP-hard to distinguish whether $W(\calV) = 0$ or $W(\calV) \geq 1/4$.
Interestingly, this result only depends on the NP-hardness of 2-2 Set Splitting.

Our results are inspired by and analogous to the hardness results for Spencer's Discrepancy Problem established by Charikar, Newman, and Nikolov~\cite{CNN11}. Spencer~\cite{spencer85} proved that for vectors $\vv_{1}, \ldots, \vv_{n}$
  in $\setof{0, 1}^{n}$, there always exists a $\xx \in \setof{\pm 1}^{n}$ such that $\norm{\sum_{i} \xx (i) \vv_{i}}_{\infty} \leq 6 \sqrt{n}$.
Charikar, Newman, and Nikolov prove that there is a constant $c$ for which it is NP-hard to distinguish vectors for which this sum can be made zero from those for which the sum always has infinity norm at least $c \sqrt{n}$. 
We follow their lead in deriving hardness from the hardness of approximating Max 2-2 Set Splitting.
However, our reduction seems very different from theirs.
For Spencer’s discrepancy problem, the NP-hardness of approximating Max 2-2 Set Splitting immediately implies that it is NP-hard to distinguish vectors $\vv_1, \ldots, \vv_n \in \setof{0,1}^n$ for which there exists an $\xx \in \setof{\pm 1}^n$ satisfying $\sum_{i=1}^n \xx(i) \vv_i = {\bf 0}$ from those for which every $\xx \in \setof{\pm 1}^n$ must have $\norm{\sum_{i=1}^n \xx(i) \vv_i}_{\infty} \ge 2$. 
The main challenge of \cite{CNN11} is amplifying this discrepancy gap from $0$ vs $2$ to $0$ vs $c \sqrt{n}$. 
For Weaver’s problem, the NP-hardness of 2-2 Set Splitting immediately implies that it is NP-hard to distinguish rank-3 matrices $A_1, \ldots, A_n$ of norm at most $1/4$ that satisfy  $\sum_{i=1}^n A_i = I$ for which there exists an $\xx \in \setof{\pm 1}^n$ satisfying $\sum_{i=1}^n \xx(i) A_i = {\bf 0}$ from those for which every $\xx \in \setof{\pm 1}^n$ must have $\norm{\sum_{i=1}^n \xx(i) A_i} \ge 1/2$. 
The first challenge in our work is that of turning these rank-3 matrices into rank-1 matrices that satisfy the conditions of Weaver's problem.
Our analog of amplification appears when we produce vectors of smaller norm.
Another difference between our results is that we do not know a polynomial time algorithm that approximately solves Weaver's problem, while Bansal~\cite{bansal10} showed that Spencer's problem could be approximately solved in polynomial time.

\section{Notation}

We write $\ee_{i}$ for the elementary unit vector with a $1$ in coordinate $i$,
  and we let $\bvec{1}$ denote the vector will all entries $1$.

As mentioned earlier, we write $\vv^{*}$ for the conjugate transpose of a vector $\vv$.
As this paper only constructs vectors with Real entries, one can just treat this as the transpose.
We let $\norm{\vv} = \sqrt{\vv^{*} \vv}$ denote the standard Euclidean norm of the vector $\vv$.
Unless otherwise specified, when we write the norm of a matrix we mean the operator norm.
We recall that the operator norm of a matrix is at least as large as the operator norm of every one of its submatrices.
For a symmetric matrix, the operator norm is the largest absolute value of its eigenvalues.

The other norm we consider of a matrix is its Frobenius norm, written $\norm{M}_{F}$, which equals the square root of the sum of the squares of the entries of $M$. 
From the identity $\norm{M}_{F}^{2} = \trace{M M^{*}}$, one can see 
  that the square of the Frobenius norm of $M$ equals the sum of the squares of the singular values of $M$.

\section{2-2 Set Splitting}

The \ttssp was defined and proved NP-complete by 
Guruswami \cite{guruswami04}.
An instance of the problem consists of a list of sets $S_{1}, \ldots, S_{m}$, each of which contains exactly four elements
  of $\setof{1,\ldots, n}$ which we identify with $\pm 1$ valued variables 
  $\xx (1), \ldots, \xx (n)$.
A vector $\xx \in \setof{\pm 1}^{n}$ \textit{satisfies} a set $S_{j}$ if $\sum_{i \in S_{j}} \xx (i) = 0$,
  and an $\xx$ satisfies the instance if it satisfies all the sets.
We say that an instance is $\gamma$-unsatisfiable if for 
  every $\xx$ at least a $\gamma$ fraction of the sets are unsatisfied. 
Guruswami proves that for every $\epsilon > 0$ it is NP-hard to distinguish 
  satisfiable instances from those that are $(1/12 - \epsilon)$-unsatisfiable.
Charikar, Guruswami, and Wirth \cite{CGW05} observe that Guruswami's construction has the property that there is a constant $B$ so that no variable appears in more than $B$ sets.

\begin{theorem}[Guruswami]
For every $\epsilon > 0$ there is a constant $B$ so that 
  for instances of the 2-2 Set Splitting Problem in which every variable appears in at most $B$ sets, it is 
  NP-hard to distinguish 
  satisfiable instances from those  $1/12 - \epsilon$ 
  unsatisfiable.
\end{theorem}

As the fact that there is a constant upper bound on the number of occurrences of each variable is not explicitly stated in \cite{guruswami04}, we sketch a simple proof with a worse constant in the Appendix.

We define the \ttsspt\ to be the restriction of the \ttssp\ to instances in which every variable appears in at most 3 sets.

\begin{lemma}\label{lem:three}
The \ttsspt\ is NP-hard. 
This remains true if we require that no pair of sets intersect in more than one variable.
Moreover, there is a constant $\gamma > 0$ such that
it is NP hard to distinguish 
satisfiable instances of the \ttsspt\ from those that are
$\gamma$-unsatisfiable.
\end{lemma}

The key to proving this lemma is the introduction of an \textit{equality gadget} that forces variables to have the same value. To force variables $a$ and $b$ to have the same value, we introduce variables
  $c$ and $y_{1}, \ldots, y_{12}$ and the seven sets 
\begin{align}
& \setof{a, y_{1}, y_{2}, y_{3}}, 
\setof{b, y_{4}, y_{5}, y_{6}}, 
\setof{c, y_{7}, y_{8}, y_{9}}, 
\setof{c, y_{10}, y_{11}, y_{12}},   
\label{eqn:eq_gadget}
\\
&
\setof{y_{1}, y_{4}, y_{7}, y_{10}}, 
\setof{y_{2}, y_{5}, y_{8}, y_{11}}, 
\setof{y_{3}, y_{6}, y_{9}, y_{12}}.  \notag  
\end{align}

\begin{figure}
\begin{center}
    \includegraphics[height=1.5in]{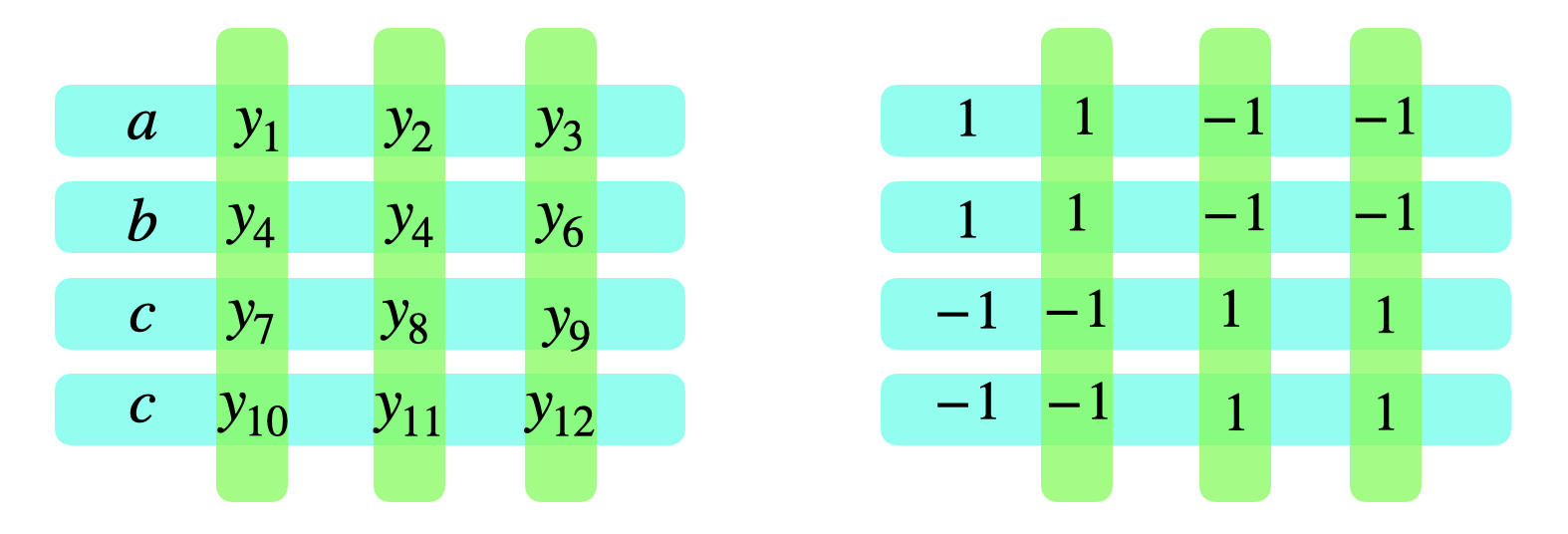}
    \end{center}
    \caption{A depiction of the equality gadget, and a setting of the variables that satisfies all the clauses when $a = b = 1$.}
    \label{fig:eq_gadget}
  \end{figure}
\begin{lemma}\label{lem:eq_gadget}
No variable appears in more than 2 of the 7 sets listed in \eqref{eqn:eq_gadget}, and variables $a$ and $b$ each appear once.
No pair of these sets intersects in more than 1 variable.
If all 7 of the sets are satisfied, then $a = b$.
And, if $a = b$ then there is a setting of the remaining variables that satisfies all the sets.
\end{lemma}
\begin{proof}
The Figure \ref{fig:eq_gadget} shows a setting of the variables that satisfies all the sets in the case that $a = b = 1$.
If $a = b = -1$, we need merely reverse all the signs.

To see that $a = b$ when these sets are satisfied, note that the last three sets require  
  half of the variables $y_{1}, \ldots, y_{12}$ to be $1$ and half to be $-1$.
If the first 4 sets are satisfied we can combine this the fact
  with the double-occurrence of $c$ to conclude that there must be an even number of $1$s among 
  $a, b, y_{1}, \ldots, y_{12}$, and so $a$ can be $1$ if and only if $b$ is as well.
\end{proof}

\begin{proof}[Proof of Lemma~\ref{lem:three}]
Let $S_{1}, \ldots, S_{m}$ be any instance of the 2-2 Set Splitting problem on 
  variables $x_{1}, \ldots, x_{n}$.
We replace each variable with many copies of itself, and use equality gadgets to force all those copies to be equal.
More formally, if variable $x_{i}$ appears $k$ times, then we create $k$ new variables
  $x_{i,1}, \ldots, x_{i,k}$, and replace each occurrence of $x_{i}$ with one of these.
Call the resulting sets on the new variables the \textit{substituted sets}.

We then add $k-1$ equality gadgets with distinct extra variables to force $x_{i,j}$ to equal $x_{i,j+1}$ for $1 \leq j < k$.
Each of the variables $x_{i,j}$ appears in at most $3$ sets:
  one substituted set and one set in each of up to two equality gadgets.
The substituted sets are all mutually disjoint, as are the equality gadgets.
The only sets that can intersect are inside equality gadgets, or a substituted set and a set in an equality gadget that both contain a variable $x_{i,j}$.
This would be the only variable in which they intersect.

The derivation of the inapproximability result uses standard techniques, such as those from \cite[Section 2.1]{guruswami04}.
Assume that the input instance is $1/13$-unsatisfiable.
As each equality gadget involves 7 sets and the number of equality gadgets is at most $4 m$, the total number of sets in the new system is at most 
 $29 m$.
For any setting of the variables $x_{i,j}$, let $u_{0}$ be the number of unsatisfied substituted sets and $u_{1}$ be the number of unsatisfied sets in equality gadgets.
Call an index $i$ inconsistent if there exist $j$ and $k$ for which $x_{i,j} \neq x_{i,k}$.
The number of inconsistent indices is at most $u_{1}$.
If all the indices were consistent, we would have $u_{0} \geq m / 13$.
As each original variable appears in at most $B$ sets, 
  the number of unsatisfied substituted sets must be at least 
  $m/13 - B u_{1}$.
Thus, the number of unsatisfied sets is at least 
\[
\max (u_{0}, u_{1}) \geq \frac{m}{13} \frac{1}{B+1},
\]
and the new instance is $\gamma$-unsatisfiable for
$ \gamma \geq  1/(13 \cdot 29 \cdot (B+1)).$
\end{proof}

\section{$\alpha = 1/4$}

\begin{theorem}\label{thm:one_quarter}
  Given a list of 
  $1/4$-Weaver vectors $\calV$, it is NP-hard to distinguish whether $W(\calV) = 0$ or $W(\calV) \geq 1/4$.
\end{theorem}

If we were considering sums of arbitrary matrices rather than sums of outer products, we could prove something like Theorem \ref{thm:one_quarter} by constructing $m$-by-$m$ diagonal matrices $D_{1}, \ldots, D_{n}$ such that 
\[
    D_{i}(j,j) = \begin{cases} 
    1/4 & \text{if $i \in S_{j}$} \\
    0 & \text{otherwise.}
    \end{cases}
\]
The corresponding 2-2 Set Splitting instance is then satisfiable if and only if there exists
 an $\xx \in \setof{\pm 1}^{n}$ so that $\sum_{i} \xx (i) D_{i} = \bvec{0}$.
 In the case where no such sum exists,  some entry of the sum must have absolute value at least $1/2$.
Note that $\sum_{i} D_{i} = I$.
To turn this problem about sums of matrices into an instance of Weaver's problem, we express each $D_{i}$ as a sum of orthogonal vectors.

Define 
\begin{equation}
\qq_{1} =
  \begin{pmatrix}
  -1/3 \\ 2/3 \\ 2/3
  \end{pmatrix},
  \quad
  \qq_{2} =
  \begin{pmatrix}
  2/3 \\ -1/3 \\ 2/3
  \end{pmatrix},  
  \quad \text{and} \quad
  \qq_{3} =
  \begin{pmatrix}
  2/3 \\ 2/3 \\ -1/3
  \end{pmatrix}.    
\end{equation}

Observe that each $q_{i}$ is a unit vector,
and that
\[
\qq_{1} \qq_{1}^{*} + \qq_{2} \qq_{2}^{*} + \qq_{3} \qq_{3}^{*} = I_{3}.
\]
We will use the following special property of these vectors.
\begin{lemma}\label{lem:Q1}
For every $\zz \in \setof{\pm 1}^{3}$ whose entries are not all equal 
  and for every diagonal matrix $X$,
  \[
  \norm{X + \sum_{i} \zz (i) \qq_{i} \qq_{i}^{*}} \geq 1.
  \]
\end{lemma}
\begin{proof}
If one of the entries of $\zz$ differs from the other two,
then the matrix $\sum_{i} \zz (i) \qq_{i} \qq_{i}^{*}$ is equal to plus or minus
 a permutation of the matrix 
 \[ R_{1} \defeq 
   I - 2 \qq_{1} \qq_{1}^{*} = 
   \frac{1}{9}
   \begin{pmatrix}
       7 & 4 & 4\\
       4 & 1 & -8\\
       4 & -8 & 1
   \end{pmatrix}.    
 \]
We now show that for every diagonal matrix $X$,
  the operator norm of $R_{1} + X$ is at least $1$.
We do this by giving a solution to the dual to the problem of choosing $X$ to minimize the operator norm.
Let
\[
  Y \defeq 
  \frac{1}{16}
  \begin{pmatrix}
    0 & 2 & 2\\
    2 & 0 & -7\\
    2 & -7 & 0
\end{pmatrix}.      
\]
The matrix $Y$ has inner product $1$ with every matrix of the form $R_{1} + X$.
The eigenvalues of $Y$ are $\lambda_{1} = -1/2$, $\lambda_{2} = 1/16$, and $\lambda_{3} = 7/16$. But, what really matters is that their absolute values sum to $1$.
Let corresponding unit-norm eigenvectors be $\pphi_{1}, \pphi_{2}, \pphi_{3}$,
and recall that $Y = \sum_{i} \lambda_{i} \pphi_{i} \pphi_{i}^{*}$
Then for every diagonal $X$,
\[
1 = \trace{(R_{1}+X)^T Y}
= 
\sum_{i} \lambda_{i} \pphi_{i}^{*} (R_{1}+X) \pphi_{i}
\leq 
\sum_{i} \abs{\lambda_{i}} \norm{R_{1}+X} 
=  \norm{R_{1}+X}.  \qedhere
\]
\end{proof}

\begin{proof}[Proof of Theorem~\ref{thm:one_quarter}]
Let $S_{1}, \ldots, S_{m}$ be an instance of the \ttsspt\ on variables 
  $x_{1}, \ldots, x_{n}$ such that no two sets intersect in more than one variable.
Lemma~\ref{lem:three} tells us that deciding whether the instance is satisfiable is NP-hard.

For each $i$ let $A_{i}$ be the sets in which variable $x_{i}$ appears.
If the variable $x_{i}$ appears in only one or two sets, introduce two or one new coordinates for that variable, and call the set of them $B_{i}$.
Otherwise, let $B_{i}$ be empty.
Let $T_{i} = A_{i} \union B_{i}$.
Define three vectors $\qq_{i,h}$ to be zero everywhere but on coordinates in $T_{i}$, on which
  they equal $(1/2) \qq_{h}$.
Let $D_{i}$ be the diagonal matrix that is $1/4$ on rows and columns indexed by $T_{i}$ and $0$ elsewhere, so that
\begin{equation}\label{eqn:1q}
    \qq_{i,1} \qq_{i,1}^{*} + \qq_{i,2} \qq_{i,2}^{*} + \qq_{i,3} \qq_{i,3}^{*} = D_{i}.     
\end{equation}
For the $i$ for which $B_{i}$ is non-empty, we introduce vectors $\rr_{j,h} = (1/2) \ee_{j,h}$ for $j \in B_{i}$ and $1 \leq h \leq 3$.
Let $B = \union_{i} B_{i}$, and let $A = \setof{1, \ldots, m}$.

We now consider Weaver's problem on the list of vectors $\calV$ consisting of  $\setof{\qq_{i,h}}$
  and $\setof{\rr_{i,h}}$.
To see that this collection of vectors is $1/4$-Weaver, first observe that each vector of form $\qq_{i,h}$ or $\rr_{i,h}$ has norm $1/2$.
Let the sum of their outer products be
\[
  M(\calV, \bvec{1}) = \sum_{1 \leq i \leq n, 1 \leq h \leq 3} \qq_{i,h} \qq_{i,h}^{*}
  +  \sum_{j \in B, 1 \leq h \leq 3} \rr_{j,h} \rr_{j,h}^{*}   .
\]
To see that $M(\calV, \bvec{1})$ is the identity, first observe that all of its off-diagonal entries are zero.
For $j \in A$, the $(j,j)$ entry is the sum of $1/4$ for every variable in set $S_{j}$, and is thus $1$.
For $j \in B$, the $(j,j)$ entry receives a contribution of $1/4$ from the sum 
  $\sum_{1 \leq h \leq 3} \qq_{i,h} \qq_{i,h}^{*}$ for the $i$ such that 
$j \in B_{i}$, and another $1/4$ from each $\rr_{j,h} \rr_{j,h}^{*}$.

Let $\zz (i,h)$ be the sign for the outer product $\qq_{i,h} \qq_{i,h}^{*}$
  and let $\ww (j,h)$ be the sign for the outer product
  $\rr_{j,h} \rr_{j,h}^{*}$, and extend the definition of $M$ so that we can write the signed sum of outer products as $M(\calV, \zz, \ww)$.

If the (3,2-2) Set Splitting instance is satisfied by $\xx$, then set 
  $\zz (i,h) = \xx (i)$ for each $i$, and for each $j$ in a non-empty $B_{i}$, 
   set $\ww (j,1) = \xx (i)$ and $\ww (j,2) = \ww (j,3) = -\xx(i)$.
This causes the signed sum of the outer products of the vectors to be the zero matrix.

If the (3,2-2) Set Splitting instance  is unsatisfiable, we will show that for every $\zz$ and $\ww$, the norm of $M(\calV, \zz, \ww)$ is 
   at least $1/4$.
We break our analysis into two cases.
In the first, we examine what happens if there exists a vector $\xx$
  so that for all $i$ and $h$,  $\zz (i,h) = \xx (i)$.
In this case, 
\[
    \sum_{1 \leq i \leq n, 1 \leq h \leq 3} \zz (i,h) \qq_{i,h} \qq_{i,h}^{*}   
    = 
    \sum_{1 \leq i \leq n} \xx (i) D_{i}.
\]
As the (3,2-2) Set Splitting instance is not satisfied by $\xx$, there must be some set $j$ for which the absolute value of sum of $\xx(i)$ for $i \in S_{j}$ is at least 2, and thus the $(j,j)$ entry of $M(\calV, \zz, \ww)$ has absolute value at least $1/2$.
As the operator norm of a matrix is at least the absolute value of its largest diagonal, in this case the norm of the signed sum must be at least $1/2$.

In the other case there is some $i$ for which
  not all of the $\zz (i, h)$ are equal.
Now, consider the entries of $M(\calV, \zz, \ww)$ that appear in rows and columns indexed by 
  $T_{i}$.
As each pair of sets $T_{i}$ and $T_{k}$ can only intersect in one element
  for $i \neq k$, the off-diagonal entries of this submatrix are equal to the 
  off-diagonals of $\sum_{1 \leq h \leq 3} \zz (i, h) \qq_{h} \qq_{h}^{*}$.
Regardless of the diagonals, Lemma~\ref{lem:Q1} tells us that this submatrix 
  has operator norm at least $1/4$, and thus $M(\calV, \zz, \ww)$ does as well.
\end{proof}

\section{General $\alpha$}

\begin{theorem}\label{thm:general}
  There exists a constant $\kappa > 0$ so that for 
  every integer $k \geq 2$
  it is
    NP-hard to distinguish a list of $1/2k$-Weaver vectors $\calW$ for which $W(\calW) = 0$
    from those for which $W(\calW) \geq \kappa / \sqrt{k}$.
  \end{theorem}

The proof employs two reductions, the first of which is a variation of the one used in the previous section.
When this reduction is applied to a $\gamma$-unsatisfiable (3,2-2) Set Splitting instance, it produces a set of vectors $\calV$ so that for all $\xx$, a constant fraction of the diagonals of $M(\calV, \xx)$ have absolute value at least $1/50$.
The second reduction converts these into instances of $1/2k$-Weaver vectors such that 
  every signed sum of those vectors has operator norm at least $\kappa / \sqrt{k}$.

In the first reduction, we use the following four orthogonal vectors:
\[
\qq_{1} \defeq \frac{1}{5}
\begin{pmatrix}
    1 \\ 4 \\ -2 \\ -2
\end{pmatrix}    
\quad
\qq_{2} \defeq \frac{1}{5}
\begin{pmatrix}
    4 \\ 1 \\ 2 \\ 2
\end{pmatrix}    
\quad
\qq_{3} \defeq \frac{1}{5}
\begin{pmatrix}
    -2 \\ 2 \\ -1 \\ 4
\end{pmatrix}    
\quad
\qq_{4} \defeq \frac{1}{5}
\begin{pmatrix}
    -2 \\ 2 \\ 4 \\ -1
\end{pmatrix}.    
\]

\begin{lemma}\label{lem:Q4}
For every $\zz \in \setof{\pm 1}^{4}$ that doesn't equal $\pm \bvec{1}$,
  for every $\ww \in \setof{\pm 1}^{3}$, and for every $1 \leq j \leq 4$,
 \[
 \abs{\sum_{i=1}^{4} (1/4) \zz (i) (\qq_{i} (j))^{2} 
 + \sum_{h = 1}^{3} (1/4) \ww (h)  } \geq 1/50.   
 \] 
\end{lemma}
\begin{proof}
The multiset of values of $(\qq_{i}(j))^{2}$ as $i$ varies from $1$ through $4$
  is $(1/25, 4/25, 4/25, 16/25)$.
Thus, every non-constant signed sum of these numbers must be an odd multiple of $1/25$ with absolute value less than $1$ and
 every non-constant signed sum of $1/4$ times these numbers 
  must have absolute value between $1/100$ and $23/100$. 
As the term $\sum_{h = 1}^{3} (1/4) \ww (h)$ can only take values in 
  $\setof{\pm 25/100, \pm 75/100}$, 
the total sum must have absolute value at least $2/100 = 1/50$.
\end{proof}

We model our first reduction on the one from the previous section, but using these vectors.
Let $S_{1}, \ldots, S_{m}$ be an instance of the (3,2-2) Set Splitting Problem on variables 
  $x_{1}, \ldots, x_{n}$.
For each $i$ let $A_{i}$ be the indices of the sets in which variable $x_{i}$ appears.
If variable $x_{i}$ appears in $k$ sets, introduce $4-k$ new coordinates for that variable, and call the set of them $B_{i}$.
As $k \leq 3$, $B_{i}$ will not be empty.
Let $T_{i} = A_{i} \union B_{i}$.
Define four vectors $\qq_{i,h}$ that are zero everywhere except on coordinates in $T_{i}$, on which
  they equal $(1/2) \qq_{h}$.
For each variable and each $j \in B_{i}$, we introduce vectors $\rr_{j,h} = (1/2) \ee_{j,h}$ for $j \in B_{i}$ and $1 \leq h \leq 3$.
Let $\calV$ consist of the vectors $\setof{\qq_{i,h}}$ and $\setof{\rr_{j,h}}$.
This collection of vectors is $1/4$-Weaver.
Let $A = \setof{1, \ldots, m}$, $B = \union_{i} B_{i}$, and note that $\sizeof{B} \leq 3m$.

\begin{lemma}\label{lem:diag1}
Let $\zz(i, h)$ be $\setof{\pm 1}$ variables for $1 \leq i \leq n$ and $1 \leq h \leq 4$.
Also let $\ww(j, h)$ be in $\pm 1$ for $j \in B$ and $1 \leq h \leq 3$.
If there are $k$ values of $i$ for which $\zz(i, h)$ is not constant over $h$,
  the matrix $M (\calV, \zz, \vv)$ must have at least 
  $k$  diagonal entries in columns in $B$ with absolute value at least $1/50$.
\end{lemma}

\begin{proof}
For every $i$ for which $\zz(i, h)$ is not constant over $h$,
  Lemma~\ref{lem:Q4} tells us that every diagonal indexed by $B_{i}$ must have absolute value
  at least $1/50$.
\end{proof}

\begin{lemma}\label{lem:reduction1}
Let $\calV$ be the vectors produced by this reduction on a 
 \ttsspt\ instance.
 Every vector in $\calV$ has at most $4$ non-zero entries,
 and no coordinate is in the support of more than $7$ of the vectors.
 If the set splitting instance is satisfiable, then $W(\calV) = 0$.
 If the set splitting instance is $\gamma$-unsatisfiable, then 
 for every $\zz$ and $\ww$, at least a $\gamma / 12$ fraction
   of the diagonal entries of $M(\calV, \zz, \ww)$ must have absolute value at least $1/50$.
\end{lemma}

\begin{proof}
If the set splitting instance is satisfiable, let $\xx$ be the vector that satisfies it. 
We then set $\zz (i,h) = \xx (i)$ for each $i$, and for each $j$ in $B_{i}$
   we set $\ww (j,1) = \xx (i)$ and $\ww (j,2) = \ww (j,3) = -\xx(i)$.
With this signing, $M(\calV, \zz, \ww)$ becomes the all-0 matrix.

Now, assume that the set splitting instance is $\gamma$-unsatisfiable.
Let $K$ be the set of $i$ for which $\zz(i,h)$ is not constant in $h$.
That is, for which there exist $h$ and $\tilde{h}$ for which $\zz(i,h) \neq \zz(i,\tilde{h})$.
Lemma~\ref{lem:diag1} tells us that at least $k = \sizeof{K}$ of the diagonals  of $M(\calV, \zz, \ww)$ indexed by $B$ have absolute value at least $1/50$.
As the dimension of $M(\calV, \zz, \ww)$ is at most $4m$, it suffices to prove that at least $(\gamma / 3) m$ of its diagonals have have absolute value at least $1/50$.

If $k \geq (\gamma / 3) m$, this finishes the proof.
If not, define $\zzhat(i, h) = \zz(i,1)$ for $i \not \in K$,
  and $1 \leq h \leq 4$,
  and set $\zzhat(i, h) = 1$ for $i \in K$.
As the set splitting instance is $\gamma$-unsatisfiable, at least $\gamma m$ of the diagonals of $M (\calV, \zzhat, \ww)$
  in columns in $A$ have absolute value at least $1/4$.
It remains to see how these diagonals change between $\zzhat$ and $\zz$.

For each $i$, the variables $\zz(i,h)$ only appear in $3$ diagonals indexed by $A$. So, $M (\calV, \zz, \ww)$ and $M (\calV, \zzhat, \ww)$ can differ in at most $3 k$ diagonals in columns in $A$.
Thus, at least $\gamma m - 3k$ diagonals of $M (\calV, \zz, \ww)$ in columns of $A$ have absolute value at least $1/4$.
In total, we find that the number of diagonals that have absolute value at least $1/50$ is at least 
\[
\gamma m - 3 k + k \geq    (\gamma / 3) m,
\]
for $k \leq (\gamma / 3) m$.
\end{proof}

For the second reduction, we employ a family of matrices constructed by projecting the signed edge-vertex adjacency matrix of a complete graph on $k$ vertices onto a $k-1$ dimensional space. 

Fix an integer $k$.
We let $\Pi$ be a $k-1$-by-$k$ matrix whose rows are an orthonormal basis of the 
  nullspace of the all-$1$ vector in $k$ dimensions.
Let $B$ be the $k$-by-$\binom{k}{2}$ matrix whose columns contain all $\binom{k}{2}$ vectors with two non-zero entries, the first of which is $1$ and the second of which is $-1$.
Our reduction uses the matrix $G \defeq \Pi B / \sqrt{k}$.

\begin{lemma}\label{lem:Gmatrix}
The matrix $G$ is a $(k-1)$-by-$\binom{k}{2}$ matrix such that
\begin{enumerate}
  \item [a.] every column of $G$ has norm $\sqrt{2 / k}$,
  \item [b.] $G G^{*} = I$, and
  \item [c.] for every $\binom{k}{2}$-dimensional square diagonal matrix $D$,
  \[ 
    \norm{G D G^{*}} \geq \frac{1}{k} \sqrt{\frac{2}{k-1}} \norm{D}_{F},
  \]
  where $\norm{D}_{F}$ is the Frobenius norm of $D$---the square root of the sum of the squares of its entries.
\end{enumerate}
\end{lemma}
\begin{proof}
As every column of $B$  has sum 0, it is orthogonal to the all-ones vector and
  so multiplying by $\Pi$ does not change its norm.
As these columns have norm $\sqrt{2}$, the columns of $G$ have norm $\sqrt{2 / k}$.

To compute $G G^{*}$, first observe that $B B^{*} = k I - J$, where $J$ is the all-ones matrix of dimension $k$.  This matrix has eigenvalue $k$ with multiplicity $k-1$ and one eigenvalue of $0$.
As $\Pi$ is a projection orthogonal to the nullspace of $B$, 
$\Pi B B^{*} \Pi$ equals $k I_{k-1}$.

Every diagonal entry of $D$ appears twice as an off-diagonal of the matrix 
  $B D B^{*}$.
One easy way to see this is to index the columns of $B$ by pairs $(i,j)$ with $i < j$,
  where column $(i,j)$ equals $\ee_{i} - \ee_{j}$.
If we index the diagonal entries of $D$ similarly and label them $d_{i,j}$ we have
\[
B D B^{*} = \sum_{i < j} d_{i,j} (\ee_{i} - \ee_{j}) ( \ee_{i} - \ee_{j})^{*}.
\]
Thus, 
\[
  \norm{B D B^{*}}_{F}^{2} \geq 2 \norm{D}_{F}^{2}.
\]
As the columns of $B$ have sum 0, they lie in the span of the rows of $\Pi$.
So,
\[
  \norm{\Pi B D (\Pi B)^{*}}_{F}^{2} = \norm{B D B^{*}}_{F}^{2},
\]
and we may conclude that
\[
  \norm{G D G^{*}}_{F}^{2} \geq \frac{1}{k^{2}} \norm{B D B^{*}}_{F}^{2}.
\]
As the Frobenius norm is the sum of the squares of the $(k-1)$ eigenvalues of 
  $G D G^{*}$,
\[
  \norm{G D G^{*}}_{2}^{2} \geq \frac{1}{k-1} \norm{G D G^{*}}_{F}^{2}
  \geq \frac{1}{(k-1) k^{2}} \norm{B D B^{*}}_{F}^{2}
    \geq \frac{2}{(k-1) k^{2}}  \norm{D}_{F}^{2}.  \qedhere
\]
\end{proof}

We now describe the second reduction.
Let $\calV$ be the set of vectors produced by the first reduction and described by Lemma \ref{lem:reduction1},
  and let $m_{1}$ be the dimension of the space in which they reside.
We now partition the coordinates of these vectors, $\setof{1, \ldots, m_{1}}$
  into at most $22$ classes so that for each vector and each class, 
  the vector has at most one non-zero entry a coordinate in that class.
To see that this is possible, and that such a partition is computable efficiently, note this this is a problem of 22-coloring a graph with maximum degree at most 21: the vertices are the coordinates, the edges go between coordinates that are in the support of the same vector, and the graph has degree at most 21.
So, a greedy coloring algorithm will do the job.
Let $C_{1}, \ldots, C_{22}$ be the classes of coordinates, and let 
  $c_{i} = \sizeof{C_{i}}$ for each $i$.

Given a choice of $k$, we would like to partition each class $C_{i}$ into sets of size $\binom{k}{2}$.
As this is not necessarily possible, for each $i$ let $a_{i}$ be the integer between $0$ and 
  $\binom{k}{2}-1$ so that $c_{i} + a_{i}$ is divisible by $\binom{k}{2}$, and let
  $a = \sum_{i} a_{i}$.
We add $a$ additional coordinates, and assign $a_{i}$ of them to class $C_{i}$ for each $i$.
Let $m_{2} = m_{1} + a$ be the number of coordinates after these are added.
We then create a new list of vectors, $\calU$ by 
\begin{itemize}
  \item embedding each vector of $\calV$ into the $m_{2}$ dimensional space by setting each extra coordinate to 0, and
  \item for each of the $a$ new coordinates, $j$, adding $4$ vectors $\rr_{j,h} = (1/2) \ee_{j}$ for $1 \leq h \leq 4$.
\end{itemize}
The list of vectors $\calU$ is $(1/4)$-Weaver.
If $W (\calV) = 0$, then $W(\calU) = 0$ as well: use the same signing for each vector derived from $\calV$,
  and then for each new coordinate $j$ assign half of the $\rr_{j,h}$ a positive sign and half a negative sign.

Now, partition each class of coordinates into groups of size $\binom{k}{2}$, and call the resulting 
  $l \defeq m_{2} / \binom{k}{2}$ classes $D_{1}, \ldots, D_{l}$.
We now describe a rectangular matrix $F$ with $m_{2}$ columns and 
  $(k-1) l$ rows.
Partition the rows of $F$ into $l$ sets of size $k-1$, which we call $E_{1}, \ldots, E_{l}$.
This partition can be arbitrary, but to ease visualization one could make each set consecutive.
We define $F$ to be zero everywhere, except on submatrices consisting of rows indexed by 
  $E_{i}$ and the columns indexed by $D_{i}$, on which it equals $G$.
The final set of vectors produced by our reduction, $\calW$, is the result of multiplying each vector in $\calU$ by $F$.
  
\begin{lemma}\label{lem:reduction2}
Let $\calV$ be the set of vectors produced by the first reduction and analyzed in Lemma~\ref{lem:reduction1}.
Let $m_{1}$ be the dimension of the space in which the vectors in $\calV$ lie, and assume that
  $m_{1} \geq 22 \binom{k}{2}$.
Let $\calW$ be the result of the second reduction.
The vectors $\calW$ are $1/2k$-Weaver.
If $W(\calV) = 0$, then $W(\calW) = 0$.
If for every $\pm 1$ vector $\xx$ at least a 
  $\phi$ fraction of the diagonals of $M(\calV, \xx)$
  have absolute value greater than $\delta$,
then 
\[
W(\calW) \geq \delta \sqrt{\frac{\phi}{2k}}
\]
\end{lemma}
\begin{proof}
We exploit the algebraic characterization of the second reduction:
\[
M(\calW, \xx) = F M(\calU, \xx) F^{*}.
\]
This immediately tells us that an $\xx$ that makes the right side zero will also make the left side zero.
It also implies that
  for every $i$ the submatrix of $M(\calU, \xx)$ indexed by rows and columns in $D_{i}$
  is diagonal.
This is because every vector in $\calU$ has at most one nonzero entry indexed by $D_{i}$,
  and the matrix $M(\calU, \xx)$ is a signed sum of outer products of vectors in $\calU$.

To see that $\calW$ is $1/2k$-Weaver, we first compute the norms of these vectors. 
As the non-zero entries of each vector in $\calU$ appear in disjoint blocks, and every column of $F$ has norm $\sqrt{2/k}$, the squared norm of $F$ times any vector in $\calU$ is $(2/k)$ times the squared norm of that vector: $(1/4)(2/k) = 1/2k$.
Also note that $F F^{*} = I$, so 
\[
  M(\calW, \bvec{1}) = F M(\calU, \bvec{1}) F^{*}
    =   F I F^{*} = I.
\]

Consider a vector
  $\xx$ for which at least a 
  $\phi$ fraction of the diagonals of $M(\calV, \xx)$
  have absolute value at least $\delta$.
Note that $a \leq 22 \binom{k}{2}$, so the assumption that 
  $m_{1} \geq 22 \binom{k}{2}$ implies $m_{2} \leq 2 m_{1}$.
This means that
  at least a $\phi / 2$ fraction of the diagonals of 
  $M(\calU, \xx)$ have absolute value at least $\delta$.
As the sets $D_{1}, \ldots, D_{l}$ partition the columns of this matrix,
  there must be some set of columns $D_{i}$
  such that at least a $\phi / 2$ fraction of the diagonals 
  in the rows and columns indexed by $D_{i}$ have absolute value at least $\delta$.
Call this submatrix $M_{i}$, and notice that it has squared Frobenius norm
  at least $\binom{k}{2} \delta^{2} \phi / 2$.
So,
\[
\norm{M(\calW, \xx)} = \norm{F M(\calU, \xx) F^{*}}
\ge \norm{G M_i G^*}
\ge \frac{1}{k} \sqrt{\frac{2}{k-1}} \norm{M_i}_F
\ge \delta \sqrt{\frac{\phi}{2k}}
\]
where the second-to-last inequality follows from part $c$ of Lemma \ref{lem:Gmatrix}.
This implies $W(\calW) \ge \delta \sqrt{\frac{\phi}{2k}} $.
\end{proof}

\begin{proof}[Proof of Theorem~\ref{thm:general}]
On input an instance of the (3,2-2) Set Splitting Problem, 
  let $\calV$ be the set of vectors produced by the first reduction,
  and let $\calW$ be the set of vectors produced by the second.
By applying Lemmas~\ref{lem:reduction1} and~\ref{lem:reduction2}, 
  we see that if the instance is satisfiable, then 
  $W(\calV) = W(\calU) = W(\calW) = 0$.
On the other hand, if the instance is $\gamma$-unsatisfiable, then 
  Lemma~\ref{lem:reduction1} implies that
  for all $\pm 1$ vectors $\zz$ and $\ww$ at least a $\phi = \gamma / 12$
  fraction of the diagonal entries of $M(\calV, \zz, \ww)$ have absolute 
  value at least $\delta = 1/50$.
Lemma~\ref{lem:reduction2}, then allows us to conclude that 
  \[ W(\calW) \geq \frac{1}{50} \sqrt{\frac{\gamma}{ 24 k}} = \kappa / \sqrt{k}, \]  
where
\[ \kappa \defeq \frac{1}{100} \sqrt{\frac{\gamma}{6}}. \]
So, the problem of distinguishing whether a (3,2-2) set splitting instance is satisfiable 
  or $\gamma$-unsatisfiable is polynomial-time reducible to the problem of distinguishing 
  a set of $1/2k$-Weaver vectors $\calW$ with $W(\calW) = 0$ from a set for which 
  $W(\calW) \geq \kappa /  \sqrt{k}. $
\end{proof}

We remark that this construction can be carried out whenever the original (3,2-2) Set Splitting instances has a number of sets that exceeds $22 \binom{k}{2}$.
This will result in a number of vectors that is a most a constant times the number of sets.
Thus, we only require that $k$ be at least some constant times the square root of the number of vectors.

\section{Remarks}
We first emphasize that our hardness results do not say that it is hard to find an $\xx$ 
  giving an operator norm at or above the guarantee provided by \cite{marcusSS15,BCMS19}.
We only prove that it is hard to improve on this guarantee by a constant factor.

The original form of Weaver's conjecture $KS_{2}$ states that there exist constants 
  $\alpha > 0 $ and $\beta < 1$ 
  such that for vectors $\vv_{i}$ of norm at most $\sqrt{\alpha}$ whose outer products have sum 
  with operator norm less than $1$, there exists a partition of those vectors into two sets so that 
  in each set the sum of the outer products has operator norm at most $\beta$.
These vectors could differ from those in $\alpha$-Weaver position in that the sum of their outer products does not need to equal the identity.
Weaver proved that the conjecture is unchanged if one requires the sum of the outer products of the vectors to be the identity.
Instead of considering the sum of the outer products in each set, we consider the difference of the sum of the outer products by assigning a $+1$ to every vector in one set and a $-1$ to every vector in the other.
However, when we consider such signed sums the condition that the sum of the outer products is the identity is no longer equivalent to the condition that the sum has operator norm at most $1$.
To prove an upper bound on the discrepancy for vectors of bounded norm whose sum of outer products has operator norm at most $1$, one can use the results of Kyng, Luh, and Song \cite{KLS}.
Instead of outer products of vectors, Cohen~\cite{Cohen16} and Br{\"a}nd{\'e}n~\cite{branden18} have shown that it is possible to prove analogous discrepancy results for sums of positive semidefinite matrices of bounded trace.

One may wonder what to make of our results when the vectors $\calW$ produced are not rational,
  because it is not clear that they can be represented exactly, and thus their representation in floating point might not be precisely $\alpha$-Weaver.
One way to fix this is to round them to floating point numbers, and then apply a linear transformation that forces the sum of their outer products to be the identity.
If done with enough precision, this will cause their norms to increase negligibly.
We also observe that the vectors can be made rational whenever $k$ is a square.
If $k = s^{2}$, then one can choose $\Pi$ to be
 the horizontal concatenation of the vector $-\bvec{1}_{k-1} / s$
 with the matrix $I_{k-1} - J_{k-1} (s / (k(s+1)))$, where $J_{k-1}$
 is the $(k-1)$-dimensional square matrix with all entries $1$.

\bibliographystyle{alpha}
\bibliography{pz_ref}

\appendix 

\section{Hardness of 2-2 Set Splitting}

The purpose of this section is to sketch a simple proof that it is NP-hard to distinguish satisfiable (3,2-2) set splitting instances from $\gamma$-unsatisfiable ones, for some constant $\gamma > 0$.

We first sketch a proof that 2-2 Set Splitting is NP-hard.
We then explain why it is hard to distinguish satisfiable instances 
from $\gamma$-unsatisfiable ones, for some constant $\gamma > 0$,
  even when each variable appears in at most a constant number of sets.
  
Our notation follows that of H{\ao}stad~\cite{haastad01} and Guruswami~\cite{guruswami04}. 
Whereas the purpose of those papers is to obtain tight hardness of approximation results, our purpose in this appendix is just to obtain simple proofs of hardness up to some constant.

We begin by recalling the NP-hardness of E3-SAT: 3-SAT in which every clause contains exactly 3 distinct variables.
We will reduce this to NAE-E3-SAT, where we recall that 
  the NAE-SAT problem consists of not-all-equal clauses that are satisfied when their terms are not all equal, and NAE-E$k$-SAT is the restriction of NAE-SAT to instances in which every clause contains exactly $k$ distinct variables.

The standard reduction from E3-SAT to NAE-E4-SAT is obtained by creating one extra variable, $z$, and replacing every clause in a SAT instance with an NAE clause that contains the same terms along with $z$.
If the SAT instance is satisfied by an assignment $\xx$, then the NAE-SAT instance is satisfied by the same assignment and $z$ set to false.
Conversely, observe that satisfying assignments of NAE-SAT instances remain satisfying if one negates all the variables.
So, if the NAE-E4-SAT instance is satisfiable, we may assume that $z$ is false and that the remaining variables provide a satisfying assignment to the SAT instance.

We then reduce the NAE-E4-SAT instance to an NAE-E3-SAT instance by splitting up each NAE clause. For each NAE clause with terms $t_{1}, t_{2}, t_{3}, t_{4}$, we introduce one new variable $y$, and then replace the clause with two clauses:
  one with terms $t_{1}, t_{2}, y$ and one with $\overline{y}, t_{3}, t_{4}$.
  
To reduce NAE-E3-SAT to 2-2 Set Splitting, we first show that we can reduce it to an NAE-E3-SAT problem in which no variable is negated in any NAE clause.
We may accomplish this by introducing a gadget that forces variables to be the negations of each other.
To force variables $x$ and $y$ to be negations of each other, we introduce extra variables 
  $a$, $b$, and $c$, and include the NAE-E3 clauses
  \[
  (x, y, a), (x, y, b), (x, y, c), (a, b, c).
  \]
If $x$ is the negation of $y$, then these clauses are satisfied by any choice of $a$, $b$, and $c$ that are not all equal. Conversely, if $a$, $b$, and $c$ are not all equal then these clauses can only be satisfied if $x$ differs from $y$.

Finally, we may reduce NAE-E3-SAT to 2-2 Set Splitting by appending one extra variable to every clause, and making the result a set to be split.

If every variable occurs at most a constant number of times in the original E3-SAT instance, then every variable will occur at most a constant number of times in the 2-2 Set Splitting instance, except for the variable $z$ which was added in the reduction from E3-SAT to NAE-E4-SAT.
To fix this, we replace the variable $z$ with many variables, and then force them all to be equal. In particular, if the E3-SAT instance has $m$ clauses, we introduce variables $z_{1}, \ldots, z_{m}$, and add one to each clause to create an NAE-E4-SAT clause.
We must then introduce gadgets that force those variables to be equal.
For simplicity, for each $1 \leq j < m$, we could introduce a new variable $w_{j}$,
  and include the NAE-E3 clauses that force $z_{j} \neq w_{j}$ and $w_{j} \neq z_{j+1}$.
Of course, we do not need to split these clauses when we reduce the other NAE-E4-SAT clauses to NAE-E3-SAT clauses.

To preserve constant-factor unsatisfiability, we add more constraints than this to the variables 
  $z_{1}, \ldots, z_{m}$.
First, we recall the formulation by H{\ao}stad~\cite[Theorem 2.24]{haastad01} of one of the main results of 
  Arora \textit{et. al.} \cite{almss}:

\begin{theorem}
There exists a constant $c > 0$ such that it is NP-hard to 
distinguish a satisfiable E3-SAT instance in which every variable appears in at most 5 clauses
  from one that is $c$-unsatisfiable.
\end{theorem}

The reductions we have described so far convert satisfiable E3-SAT instances to satisfiable 2-2 set splitting instances, and they ensure that if each variable appears at most 5 times in the original instance, then each variable appears in at most a constant number of sets in the set splitting instance. To make sure that each $c$-unsatisfiable E3-SAT instance is converted into a $c'$-unsatisfiable NAE-E3-SAT instance, we impose equality relations between $z_{1}, \ldots, z_{m}$ in the pattern of an expander graph.

For example, we could use a 4-regular Ramanujan graph~\cite{Margulis,LPS} 
  on $m$ or slightly more than $m$ vertices.
If the graph has exactly $m$ vertices, then for every edge $(i,j)$ in the graph, we
  use the gadgets described above to force $z_{i} = z_{j}$.
If the graph has more than $m$ vertices, when we introduce even more copies of $z$ so that we have one for each vertex, and then proceed as before.
The gadgets ensure that if more than $k$ of the copies of $z$ differ from the majority, then at least $\omega k$ of the clauses in the gadgets will be unsatisfied, for some $\omega > 0$.
If the E3-SAT instance is $c$-unsatisfiable and only a small enough fraction of the copies of $z$ disagree with the majority, then some constant fraction of the other NAE-E3-SAT clauses must be unsatisfied.

As the other parts of the reduction only involve a constant number of locally substituted clauses, we may prove as in Lemma~\ref{lem:three} that the $c$-unsatisfiable E3-SAT instances become constant-unsatisfiable 2-2 set splitting instances. 
As each variable appears in at most a constant number of sets, we can then use 
  Lemma~\ref{lem:three} to ensure that each variable occurs in at most three sets.

\end{document}